\documentclass[twocolumn,10pt]{article}
\usepackage{epsfig}
\usepackage{amssymb}
\usepackage{amsmath}
\usepackage{amsfonts}
\usepackage{amsthm}

\newtheorem{theorem}{Theorem}
\newtheorem{definition}{Definition}

\advance\oddsidemargin by -1in  
\advance\evensidemargin by -1in 
\advance\topmargin by -1in      
\begin{document}
\title{More about Base Station Location Games }

\author{Fran\c cois M\'eriaux, Samson Lasaulce, Michel Kieffer\\
L2S - CNRS - SUPELEC - Univ Paris-Sud\\
3 rue Joliot-Curie
F-91192 Gif-sur-Yvette, France\\
\{meriaux,lasaulce,kieffer\}@lss.supelec.fr}

\maketitle
\small
\begin{abstract}

This paper addresses the problem of locating base stations in a
certain area which is highly populated by mobile stations; each
mobile station is assumed to select the closest base station. Base
stations are modeled by players who choose their best location for
maximizing their uplink throughput. The approach of this paper is to
make some simplifying assumptions in order to get interpretable
analytical results and insights to the problem under study.
Specifically, a relatively complete Nash equilibrium (NE) analysis
is conducted (existence, uniqueness, determination, and efficiency).
Then, assuming that the base station location can be adjusted
dynamically, the best-response dynamics and reinforcement learning
algorithm are applied, discussed, and illustrated through numerical
results.

\end{abstract}

\section{Introduction}

Mobile terminals (MTs) are currently gaining increased autonomy of
decision to allow a better use of the available wireless resources.
For example, MTs may choose their wireless access technology or the
base station (BS) or access point to which they want to connect. We
could imagine that this may be done in the future independently of
the network operator owner of the BS. A mobile operator deploying
BSs for a wireless network will have to deal with these new
characteristics. If his goal is to maximize the traffic gathered by
his own BSs, he will have to take into account the presence of
competitor network operators when deciding on the location of BSs.
If every operator involved has the same reasoning, this problem of
BSs placement may be cast in the framework of game theory and more
precisely in the context of location games.


The history of location games starts with the work of Hotel\-ling~\cite{Hotelling1929}
in which the notion of spatial competition in a
duopoly situation is introduced. More precisely, two firms compete for
benefits over a finite segment crowded with customers. This results in the
partition of the segment into a convex area of influence for each firm.
Plastria~\cite{Plastria01} gives an overview of optimization approaches to
place new facilities in an environment with pre-existing facilities. A large
overview on location games is also presented in~\cite{GabThi92}. Location
games are extended to the context of wireless networks with works such as~%
\cite{Altman09} and~\cite{SilvaDownlnk10}. The main difference
arising in this new context is the interaction between MTs due to
the mutual interference they generate. This point makes the
association problem between MTs and BSs complex. As an association
between a MT and a BS depends on SINR, the association relies on the
respective locations of the MT and the BS, but also on the MTs
already connected to the BS. Whereas~\cite {SilvaDownlnk10} focuses
on the downlink case, in~\cite{Altman09} the location of BSs and the
association choice of the MTs is treated as a Stackelberg
game~\cite{Stack} in the uplink case. The context of our work is
similar to the one in~\cite{Altman09} but several interesting
results are obtained in the present paper. The main
\emph{contributions} of this paper can be summarized as follows:

\begin{itemize}
  \item As in \cite{Altman09} MTs are assumed to operate
  in the uplink and to be distributed along a one-dimensional
  region. However, each MT
   is assumed to select the closest BS (e.g., based on
   measure given by a GPS -global positioning system- receiver).
   This leads to a convenient form for the BSs utility
   functions (Sec.~\ref{sec:uti}). As a consequence, the existence of a pure Nash equilibrium
   can be made rigorously (Sec.~\ref{sec:exi}). 

  \item Due to the symmetry of the problem, multiple Nash equilibria
  generally exist. However, if the locations can be ordered (which
  is easy for one-dimension regions), the Nash equilibrium can be
  determined and checked to be unique (Sec.~\ref{sec:nash_carac}).

  \item By making the reasonable assumption that the BS
  heights are much less than the typical distance between the BSs, the game can be further simplified and shown to be
  a form of Cournot oligopoly~\cite{Cournot1838}.

  \item In the two-player case, the efficiency of the Nash equilibrium is studied by
  evaluating the price of anarchy~\cite{Koutsoupias99worst-caseequilibria}. The
  influence of deploying its BS in the first place is
  studied by considering a Stackelberg formulation of the
  problem (Sec.~\ref{sec:Stackelberg}).

  \item Assuming that the BS locations can be adjusted
  dynamically (which would be relevant in scenario like the one of
  small cells where only some of the small BS have to be
  active), the best response dynamics and reinforcement learning
  algorithm~\cite{BushMosteller55}\cite{Sastry1994}\cite{LearningAutomtaNarendra89} are performed (Sec.~\ref{sec:dynamic}).

  \item The made assumptions lead to several interpretations which could be
  further analyzed in the light of a more general framework (e.g.,
  in two-dimensional regions).
\end{itemize}


The remainder of the paper is organized as follows. Section~\ref{sec:model}
introduces the physical model and the parameters of the $K$-player game.
Section~\ref{sec:Nash} describes the Nash equilibrium of the game in the
one-dimensional case. In Section~\ref{sec:comparison}, the Nash equilibrium,
the Stackelberg equilibrium, and the social optimum are compared. Section~%
\ref{sec:dynamic} presents a way to reach equilibrium using best-response dynamics and reinforcement
learning. Finally, Section~\ref{sec:concl}
concludes this work. 

\section{Model}

\label{sec:model}


Consider a plane to which a frame $\mathbb{R}$ is attached. A MT
$X$ located in a position $x \in \mathbb{R}^2$ in this plane is linked with a BS
$X_{1}$ of height $\varepsilon $ situated in $x_{1}\in \mathbb{R}^2$, see
Figure~\ref{Flo:epsilon}. We define the Signal to Noise Ratio (SNR) and the Signal to Interference plus Noise Ration (SINR) of this MT
\begin{equation}
SNR_X = \frac{P_{X}.h_{X_{1}}(x)}{\sigma_{X_{1}}^{2}},
\end{equation}
\begin{equation}
SINR_X = \frac{P_{X}.h_{X_{1}}(x)}{\sigma_{X_{1}}^{2}+I_{X_{1}}(x)},
\end{equation}
where $P_{X}$ is the transmission power of the MT $X$, i.e. the level of power chosen by the MT to transmit its signal. $\sigma
_{X_{1}}^{2}$
is the power of the channel noise, $I_{X_{1}}(x)$ is some interference term, $%
h_{X_{1}}(x)$ is the attenuation introduced by the uplink channel from $X$
to $X_{1}$. Here, it is assumed that
\begin{equation}
h_{X_{1}}(x)=\left( \left| x-x_{1}\right| ^{2}+\varepsilon ^{2}\right) ^{-%
\frac{\alpha }{2}},  \label{Eq:ChannelAttn}
\end{equation}
where $\left| x\right| $ is the $\ell _{2}$-norm of $x$ and $\alpha \geq 2$
is the path-loss exponent, $\alpha =2$ corresponding to the free-space
path-loss case. A higher value of $\alpha$ suits to worse channel conditions.

With Single User Decoding (SUD) at the BS $X_{1}$, there is no hierarchy for decoding the incoming signals at the BS. Hence, the signal from MT $X$ is decoded by taking into account the full interference and the uplink capacity between $X$ and $X_{1}$ may be written as
\begin{equation}
C_{X}=\log \left( 1+SINR_X \right), \label{Eq:Capa}
\end{equation}

Without loss of generality, when several MTs are considered, it is assumed that $P_{X}$ does not
depend on the MT and is normalized, \emph{i.e.}, $P_{X}=1$.
Moreover, the channel conditions, described by
\eqref{Eq:ChannelAttn}, are the same for every MT and the noise
power is constant $\sigma _{X_{1}}=\sigma $. With these assumptions,
\eqref{Eq:Capa} becomes
\begin{equation}
C_{X}=\log \left( 1+\frac{\left( \left| x-x_{1}\right| ^{2}+\varepsilon
^{2}\right) ^{-\frac{\alpha }{2}}}{{\sigma }^{2}+I_{X_{1}}(x)}\right).
\end{equation}

\begin{figure}[tbp]
\centering
\includegraphics[scale=0.3]{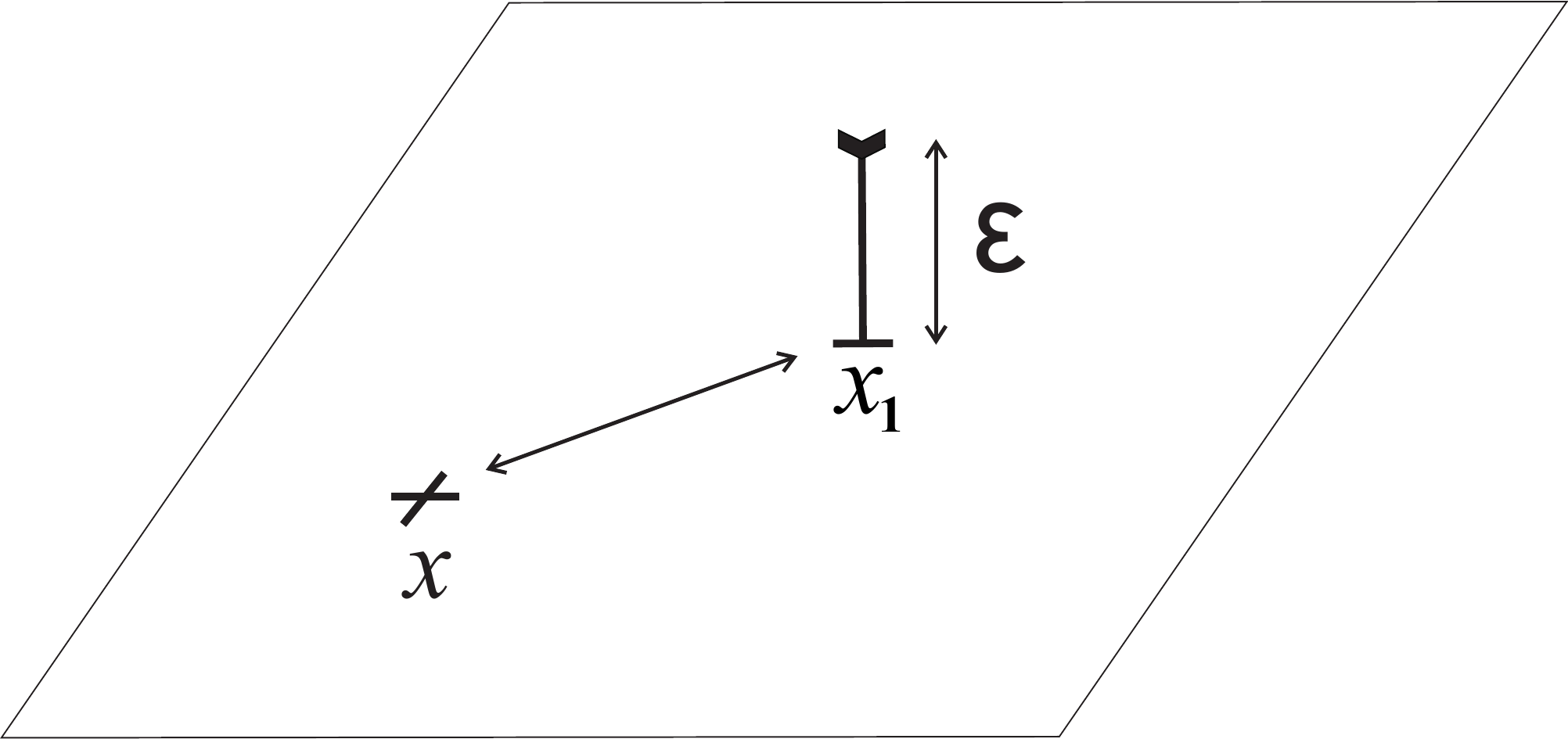}
\caption{Base station located at $x_1$ and mobile station located at
$x$} \label{Flo:epsilon}
\end{figure}
When several BSs are located on the plane, the MTs are assumed to be
able to choose the BS they want to be linked with. In this paper, we
consider that
this association is made based on SNR. Given $\left( \ref{Eq:ChannelAttn}%
\right) $, choosing the BS with the highest SNR is equivalent to choosing
the closest BS. 
Thus, it is assumed that a MT always chooses the closest BS to its
location, leading to convex cells for BSs.
For example, considering two BSs $X_{1}$ and $X_{2}$ at positions
$x_{1}$ and $x_{2}$ and a MT $X$ at $x$, one has
\[
C_{X}=\left\{
\begin{array}{l}
\log \left( 1+\frac{(\left| x-x_{1}\right| ^{2}+\varepsilon ^{2})^{-\frac{%
\alpha }{2}}}{{\sigma }^{2}+I_{1}(x)}\right) \mbox{if }\left| x-x_{1}\right|
^{2}\leqslant \left| x-x_{2}\right| ^{2} \\
\log \left( 1+\frac{(\left| x-x_{2}\right| ^{2}+\varepsilon ^{2})^{-\frac{%
\alpha }{2}}}{{\sigma }^{2}+I_{2}(x)}\right) \mbox{if }\left| x-x_{2}\right|
^{2}\leqslant \left| x-x_{1}\right| ^{2}.
\end{array}
\right.
\]

\subsection{Base station utility}

\label{sec:uti}

The utility of a BS is taken as the sum of uplink capacities it
offers to its connected users. We assume that the number of MT is
large enough to be represented by a continuous distribution $\rho
\left( x\right) $. This assumption allows to get the utility for the
$k$-th BS as a continuous sum of the MTs uplink capacities
\begin{equation}
U_{k}(\underline{x})=\int_{S_{k}(\underline{x})}\rho (z)\log \left( 1+\frac{%
(\left|z-x_{k}\right|^{2}+\varepsilon ^{2})^{-\frac{\alpha }{2}}}{{\sigma }^{2}+I_{k}(\underline{x})}%
\right) \text{d}z,  \label{EqUtility1}
\end{equation}
where $S_{k}\left( \underline{x}\right) $ is the subset of the plane where MTs are linked with the $k$-th BS and
$\underline{x}=~^t(x_{1},x_{2},\ldots ,x_{K})$ is the vector of
locations for the set of BSs. This paper will consider only uniform
MT distribution with $\rho (x)=1$.

When considering interferences, a worst-case scenario is considered:
there is no mechanism such as beamforming~\cite{Litva96} to lower
their effects. We consider that there is no interference between MTs
connected with different BSs because of frequency reuse. Then only
interference between MTs of a same BS has to be considered. This
framework is quite similar to the one of~\cite {Altman09}, where two
competing BSs are assumed to use different frequency bands. In our
case, we consider $K$ BSs (with $K\geqslant 1$) and each of them
uses its own frequency band.

Performing SUD at the $k$-th BS, one gets
\begin{equation}
I_{k}(\underline{x})=\int_{S_{k}(\underline{x})}\left( \left| z-x_{k}\right|
^{2}+\varepsilon ^{2}\right) ^{-\frac{\alpha }{2}}\text{d}z.
\label{Eq:Interference}
\end{equation}
Utility of $k$-th BS $\left( \ref{EqUtility1}\right) $ then becomes
\begin{equation}
U_{k}(\underline{x})=\int_{\mathcal{S}_{k}(\underline{x})}\log \left( 1+%
\frac{(\left| z-x_{k}\right| ^{2}+\varepsilon ^{2})^{-\frac{\alpha }{2}}}{{%
\sigma }^{2}+\int_{\mathcal{S}_{k}(\underline{x})}(\left|
z^{\prime }-x_{k}\right|^{2}+\varepsilon ^{2})^{-\frac{\alpha }{2}}\text{d}z^{\prime }}\right) \text{d}z.
\label{EqUtility2}
\end{equation}


\subsection{Utility approximation}


When considering the low SINR regime, the useful power of each MT is
small compared to the interference term in $\left(
\ref{EqUtility2}\right) $. This is especially true for high MT
density. With this assumption, $\left( \ref {EqUtility2}\right) $
may be approximated as
\begin{equation}
\begin{aligned} U_{k}(\underline{x}) &\approx \frac{\int_{\mathcal{S}_k(\underline{x})} (\left|z-x_k\right|
^{2}+\varepsilon^{2})^{-\frac{\alpha}{2}}{\text{d}z}}{{\sigma}^{2}+\int
_{\mathcal{S}_k(\underline{x})}(\left|z^{\prime }-x_k\right|^{2}+\varepsilon^{2})^{-\frac{\alpha}{2}}{\text{d}z^{\prime }}}\\ 
&=\frac{I_k(\underline{x})}{\sigma^2 +
I_k(\underline{x})}. \end{aligned}  \label{Eq:Utility2}
\end{equation}

Note that this simplification makes the considered utility based on capacity
equivalent to a utility based on SINR such as the ones in~\cite{Altman09}.
At low-SINR regime, it is equivalent to work with a capacity-based utility
or a SINR-based utility.

Also note that one has $f\left( t\right) =t/\left( \sigma ^{2}+t\right) $ strictly
increasing over $\left[ 0,L\right] $, since its derivative is $f^{\prime
}\left( t\right) =\sigma ^{2}/\left( \sigma ^{2}+t\right) ^{2}$. Then maximizing the approximation of $U_{k}(\underline{x})$ or $I_k(\underline{x})$ is equivalent. Thus the utility we define for the game is $\widehat{U_{k}}(\underline{x})=I_k(\underline{x})$, $k\in\mathcal{K}$.

\subsection{Definition of the game}

\label{Ssec:GameDef}

In this section, we study the case of $K$ BSs competing on a segment of length $L$. Each of the BSs uses
different carrier frequencies so the MTs of different BSs do not
interfere together. As assumed in Section~\ref{sec:model}, the MT
distribution is uniform over the segment, and the set of possible
locations for the BS is $[0,L]$.
Figure~\ref{street}
illustrates the context for the two-player case.
\begin{figure}[tbp]
\centering
\includegraphics[scale=0.3]{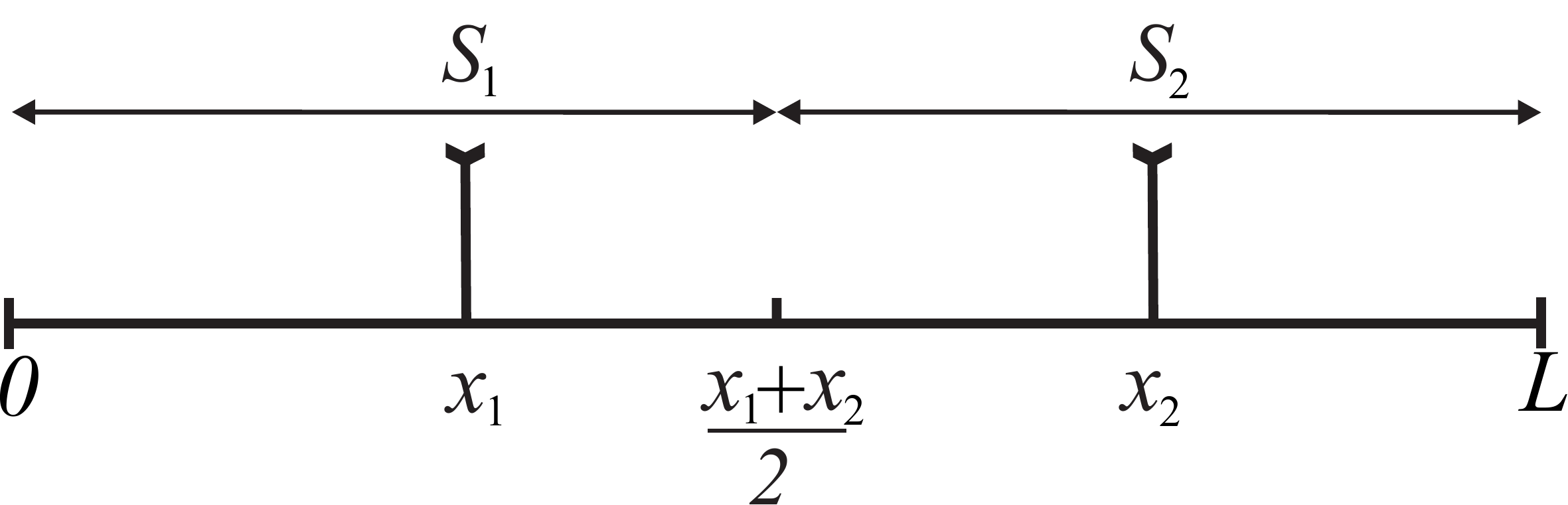}
\caption{Two base stations competing in a one-dimension space}
\label{street}
\end{figure}



\begin{definition}
\label{def:game} The strategic form of the game is given by
\[
\mathcal{G}=(\mathcal{K},{\{\mathcal{A}_{k}\}}_{k\in \mathcal{K}},{\{\widehat{U_k}\}}_{k\in \mathcal{K}})
\]
where

\begin{itemize}
\item  $\mathcal{K}=\{1,\ldots ,K\}$ is the set of players, which are here
BS.

\item  ${\{\mathcal{A}_{k}\}}_{k\in \mathcal{K}}$ is the set of actions
players can consider, here 
\begin{equation}
\label{eq:order}
\mathcal{A}_{k}=\{ x_k \in [0,L]\; |\; 0<x_{1}<\ldots <x_{K}<L\}. 
\end{equation}
Denote $\mathcal{A}=\{ \underline{x} \in [0,L]^K\; |\; 0<x_{1}<\ldots <x_{K}<L\}$.

\item  ${\{\widehat{U_k}\}}_{k\in \mathcal{K}}$ is the set of utilities players
use. 
\end{itemize}
\end{definition}
Note that we are interested in location equilibria that do not superimpose several BSs. Thus, if there exists an equilibrium, there exists a spatial order for BSs at this equilibrium, that is why we introduce this order in the action spaces $\{\mathcal{A}_k\}_{k\in\mathcal{K}}$.



\section{Nash equilibrium analysis}

\label{sec:Nash}

The aim of this section is to show the existence of a Nash equilibrium in
the location game described in Section~\ref{Ssec:GameDef} and to
characterize this equilibrium.

\subsection{Existence}
\label{sec:exi}
In this section, we focus of the existence of a Nash equilibrium in the defined game. 
To prove the existence of a Nash equilibrium, the concavity of $\widehat{U_k}(%
\underline{x})$ with respect to $x_{k}$ over $\mathcal{A}_k$, $\forall k\in \mathcal{K}$, has to
be established.

\newtheorem{lemma}{Lemma}

\begin{lemma}
\label{lemma:concavity} 
$\widehat{U_k}(\underline{x})$ is concave with respect to $x_{k}$ over $\mathcal{A}_k$, $\forall k\in \mathcal{K}$.
\end{lemma}

The proof of this lemma is in Appendix~\ref{proof:concavity}. Then one has the following theorem.

\begin{theorem}
\label{Th:Nash}In the game defined by Definition~\ref{def:game}, there
exists at least one Nash equilibrium.
\end{theorem}

\begin{proof}
\begin{itemize}
\item Using Lemma~\ref{lemma:concavity}, we know that $\widehat{U_k}\left(
\underline{x}\right) $ is concave with respect to $x_{k}$ over $\mathcal{A}_k$, $\forall k\in \mathcal{K}$,

\item $\widehat{U_k}\left(
\underline{x}\right) $ is continuous with respect to $\underline{x}$ over $\mathcal{A}$, $\forall k\in \mathcal{K}$,

\item the set of feasible actions is compact and convex for all players in the game.
\end{itemize}
The Rosen~\cite{Rosen1965} conditions for the existence of a Nash equilibrium are met and Theorem~
\ref{Th:Nash} is thus proved.
\end{proof}


\subsection{Multiplicity of NE}

Regarding to the uniqueness of the Nash equilibrium, as the
characteristics of the BSs (height and noise) are assumed to be
identical, it is interesting to note that permuting the order of BSs
leads to a symmetric system of
equation. Thus, without condition on the order of BSs as in (\ref{eq:order}%
), there are $K!$ Nash equilibria for the game and all these
equilibria are symmetric, meaning that the set of locations at
equilibrium is unique. However, if one imposes the condition order
(\ref{eq:order}), the NE can be shown to be unique by using the Diagonally Strict Concavity (DSC) condition~\cite{Rosen1965}. 

\begin{theorem}
\label{th:unicity}
In the game defined by Definition~\ref{def:game}, there
exists one single Nash equilibrium.
\end{theorem}

The proof of this theorem is in Appendix~\ref{proof:unicity}.
Having uniqueness under this order condition might seem to be a weak result in
comparison with a general uniqueness result. However, in the
framework of a dynamic process, the initial locations of the base
station might suffice to determine the effectively observed NE
(after convergence).

\subsection{Determination of the NE}

\label{sec:nash_carac}

A characterization of the equilibrium is provided in this section with
examples for small values of $K$. 
A real solution $\underline{x}$ for $\alpha >2$
has to satisfy
\begin{equation}
\label{eq:charac}
\left\{
\begin{array}{l}
x_{1}=\frac{2\sqrt{(5.2^{\frac{2}{\alpha }}-2^{\frac{4}{\alpha }%
}-4)\varepsilon ^{2}+2^{\frac{2}{\alpha }}x_{2}^{2}}-2^{\frac{2}{\alpha }%
}x_{2}}{4-2^{\frac{2}{\alpha }}}, \\
x_{k}=\frac{x_{k+1}+x_{k-1}}{2},\;\forall k\in \{2,\ldots ,K-1\}, \\
x_{K}=\frac{-2\sqrt{(5.2^{\frac{2}{\alpha }}-2^{\frac{4}{\alpha }%
}-4)\varepsilon ^{2}+2^{\frac{2}{\alpha }}(L-x_{K-1})^{2}}-2^{\frac{2}{%
\alpha }}x_{K-1}+4L}{4-2^{\frac{2}{\alpha }}}.
\end{array}
\right.
\end{equation}
See Appendix~\ref{proof:charac} for detailed derivations. In the case $%
\alpha =2$, one obtains
\begin{equation}
\label{eq:charac2}
\left\{
\begin{array}{l}
x_{1}=\sqrt{2\varepsilon ^{2}+2x_{2}^{2}}-x_{2}, \\
x_{k}=\frac{x_{k+1}+x_{k-1}}{2},\;\forall k\in \{2,\ldots ,K-1\}, \\
x_{K}=-\sqrt{10\varepsilon ^{2}+2(L-x_{K-1})^{2}}-x_{K-1}+2L.
\end{array}
\right.
\end{equation}

In the $2$-player game with $\alpha =2$, one gets
\begin{equation}
\left\{
\begin{array}{l}
x_{1}^{\text{ne}}=L-\frac{1}{2}\sqrt{2L^{2}-4\varepsilon ^{2}} \\
x_{2}^{\text{ne}}=\frac{1}{2}\sqrt{2L^{2}-4\varepsilon ^{2}}
\end{array}
\right.
\end{equation}

Figure~\ref{fig:nashequi} illustrates equilibria for $\alpha =2$ and $\alpha
=3$ for respectively two, three, and four BSs.

\begin{figure}[tbp]
\centering
\includegraphics[scale=0.65]{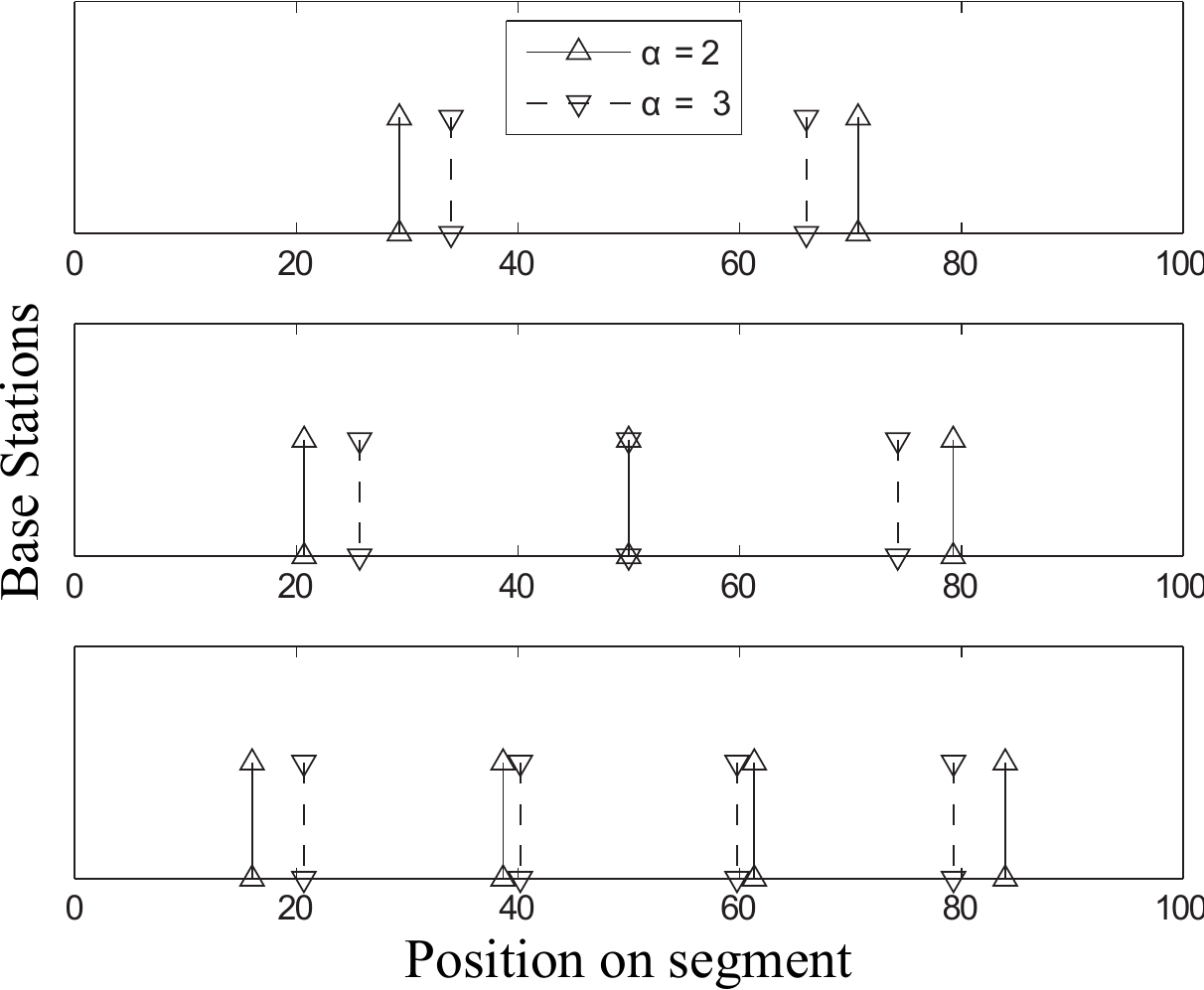}
\caption{Base Stations at Nash equilibrium for $\alpha = 3$ (dotted line)
and $\alpha = 2$ (continuous line). From the top to the bottom,
configuration with 2 BSs, 3 BSs and 4 BSs.
Parameters are $\varepsilon = 0.1$ and $L=100$.}
\label{fig:nashequi}
\end{figure}

\section{Comparing equilibria}

\label{sec:comparison}

This section compares the Nash equilibrium described in Section~\ref
{sec:Nash}, the Stackelberg equilibrium, and the social optimum.

\subsection{Stackelberg equilibrium}

\label{sec:Stackelberg}

As written in Section~\ref{sec:Nash}, the pure Nash equilibrium needs the
players to know their spatial order to be reached. If players do not know
this order but play in a chronological order, the problem changes. If the
first player plays alone at the first stage of the game knowing that other
players will place their BSs after, it turns into a Stackelberg game~\cite
{Stack} with the first player being the leader of the game. This idea is
illustrated with a two-player game with one leader (BS $1$) and one follower
(BS $2$). The leader chooses its position knowing that the follower will
place itself after. Both BSs still want to maximize their utilities and BS $%
1$ knows this point. Then, BS $1$ knows how BS $2$ is going to be placed
regarding to its own position, it is simply the best response of BS $2$
\[
x_{2}(x_{1})=\frac{-2\sqrt{(5.2^{\frac{2}{\alpha }}-2^{\frac{4}{\alpha }%
}-4)\varepsilon ^{2}+2^{\frac{2}{\alpha }}(L-x_{1})^{2}}-2^{\frac{2}{\alpha }%
}x_{1}+4L}{4-2^{\frac{2}{\alpha }}}.
\]

Then BS $1$ places itself at a location $x_{1}$ solution of
\begin{equation}
\frac{\partial \widehat{U_1}}{\partial x_{1}}(x_{1},x_{2}(x_{1}))=0.
\end{equation}

With $\alpha =2$ and neglecting $\varepsilon $ compared to every other
lengths of the problem, one gets
\begin{equation}
x_{1}^{\text{se}}=\biggl(1-\sqrt{2}+\sqrt{2-\sqrt{2}}\biggr)L,
\end{equation}
and
\begin{equation}
x_{2}^{\text{se}}=\biggl((\sqrt{2}-1)(1+\sqrt{2-\sqrt{2}})\biggr)L.
\end{equation}

On Figure~\ref{fig:stacknash}, we compare the utilities of the leader and the follower of the Stackelberg game versus the utilities of the Nash equilibrium. As we see, the leader of the Stackelberg game has a better utility than what he would have get at the Nash equilibrium. On the contrary, the follower has a worst utility. Hence, in a mobile operator point of view, it is more interesting to deploy its BS first.

\begin{figure}[tbp]
\hspace{-0.5cm}
\includegraphics[scale=0.45]{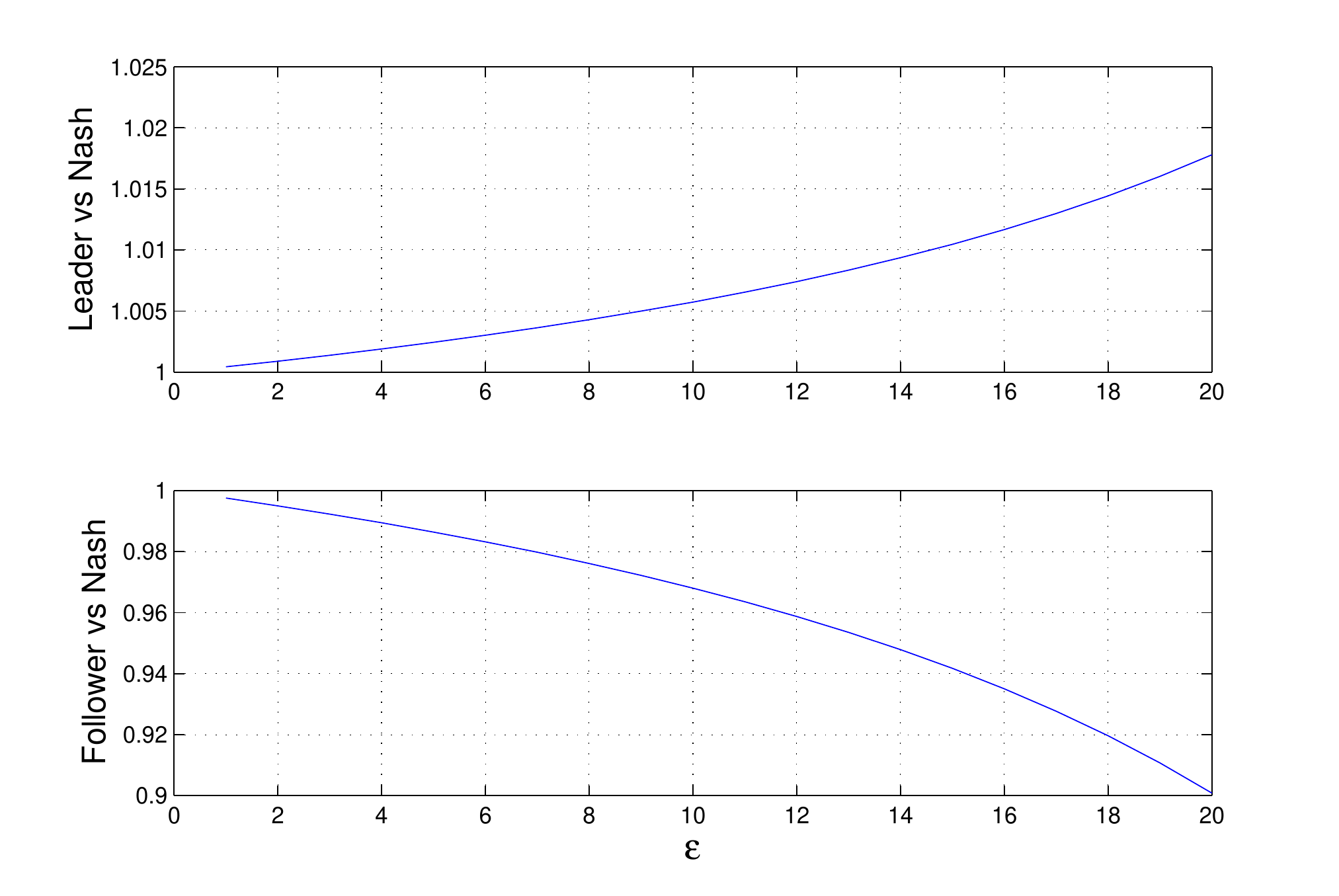}
\caption{$\frac{\widehat{U_1}^{\text{SE}}}{\widehat{U_1}^{\text{NE}}}$ and $\frac{\widehat{U_2}^{\text{SE}}}{\widehat{U_2}^{\text{NE}}}$ with regard to $\varepsilon$.}
\label{fig:stacknash}
\end{figure}

\subsection{Social optimum}

Sections~\ref{sec:Nash} and~\ref{sec:Stackelberg} provide equilibria
corresponding to situations where the BSs only consider what is best
for themselves. But these equilibria are not the best for the MTs.
In the MTs
point-of-view, the utility to consider is the complete utility sum or \emph{%
social utility}. In a $2$-BS case
\begin{equation}
\widehat{U}^{\text{so}}\left( x_1,x_2\right) =\widehat{U_1}\left( x_1,x_2\right) +
\widehat{U_2}\left( x_1,x_2\right) . \label{Eq:SocialUtility}
\end{equation}

For this utility, we know that there exists an optimum since the strategy
set $\mathcal{A}$ is compact and $\widehat{U_1}\left( x_1,x_2\right) + \widehat{U_2}\left( x_1,x_2\right)$ is continuous with respect to $(x_1,x_2)$ over $\mathcal{A}$. At
this optimum, it is proven \cite{Altman09} that

\begin{itemize}
\item  {$(i)$ the BSs place themselves at the middle of their
associated subset of segment (respectively $\mathcal{S}_1$ and $\mathcal{S}_2$),}

\item  {$(ii)$ the frontier between the two BSs is the middle of the segment.%
} 
\end{itemize}

Thus, we have the optimum
\begin{equation}
x_{1}^{\text{so}}=\frac{L}{4}\mbox{, }x_{2}^{\text{so}}=\frac{3L}{4}.
\end{equation}

Figure~\ref{fig:Nash_stack_opti} shows the Nash equilibrium, the Stackelberg
equilibrium, and the social optimum for $\alpha =2$. The locations of BSs
for Nash equilibrium and social optimum are symmetric with respect to the
middle of the segment $\left[ 0,L\right] $ whereas this is not the case for
the Stackelberg equilibrium.
\begin{figure}[tbp]
\centering
\includegraphics[scale=0.42]{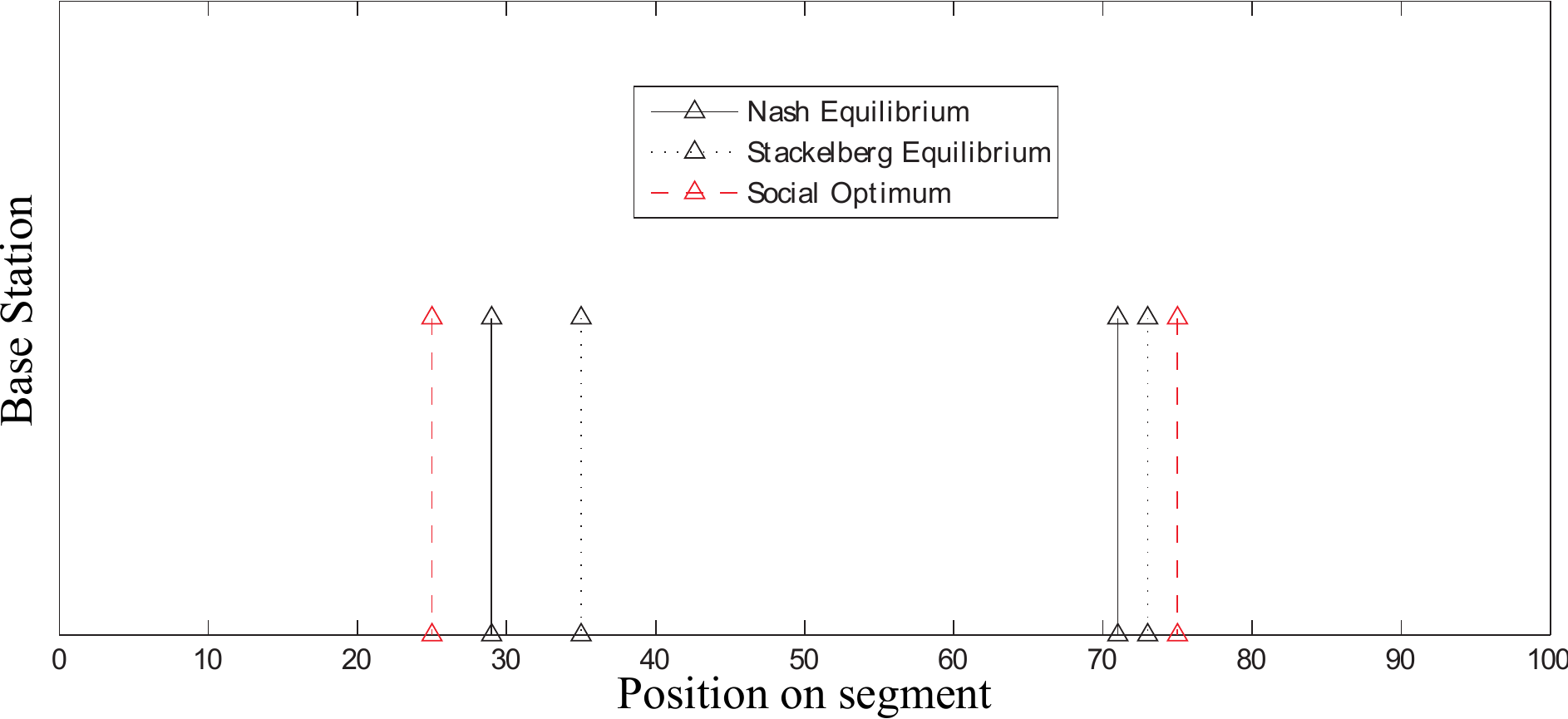}
\caption{Nash Equilibrium, Stackelberg Equilibrium, and Social Optimum for $%
\alpha =2$.}
\label{fig:Nash_stack_opti}
\end{figure}

\subsection{Price of anarchy}

To compare the Nash and the Stackelberg equilibrium in the two-player case,
we look at the utilities at equilibrium. Without any pricing mechanism, the
Nash equilibrium and the social optimum are very close in terms of
locations. As a result, they are also close in terms of utility sum. The
\emph{Price of Anarchy} (PoA), introduced in~\cite{Koutsoupias99worst-caseequilibria}, is an adequate
metric to compare these sums
\begin{equation}
PoA(eq)=\frac{\max_{\underline{x}\in \mathcal{A}}\sum_{k\in \mathcal{K}}\widehat{U}_{k}(\underline{x})}{\sum_{k\in
\mathcal{K}}\widehat{U}_{k}(\underline{x}^{\text{eq}})}.
\end{equation}

Note that the PoA is always stronger than $1$ and if the PoA is high, it means that the corresponding equilibrium is not that efficient in terms of utility sum. On the contrary, if the Po1 is close to one, the corresponding equilibrium is satisfying.
Figure~\ref{PoA} illustrates the behavior of the PoA for the Nash
equilibrium and the PoA for the Stackelberg equilibrium as a function of $%
\varepsilon $ for $2$-player case. It appears that the Stackelberg equilibrium is less efficient than the Nash equilibrium.


\begin{figure}[tbp]
\centering
\includegraphics[scale=0.48]{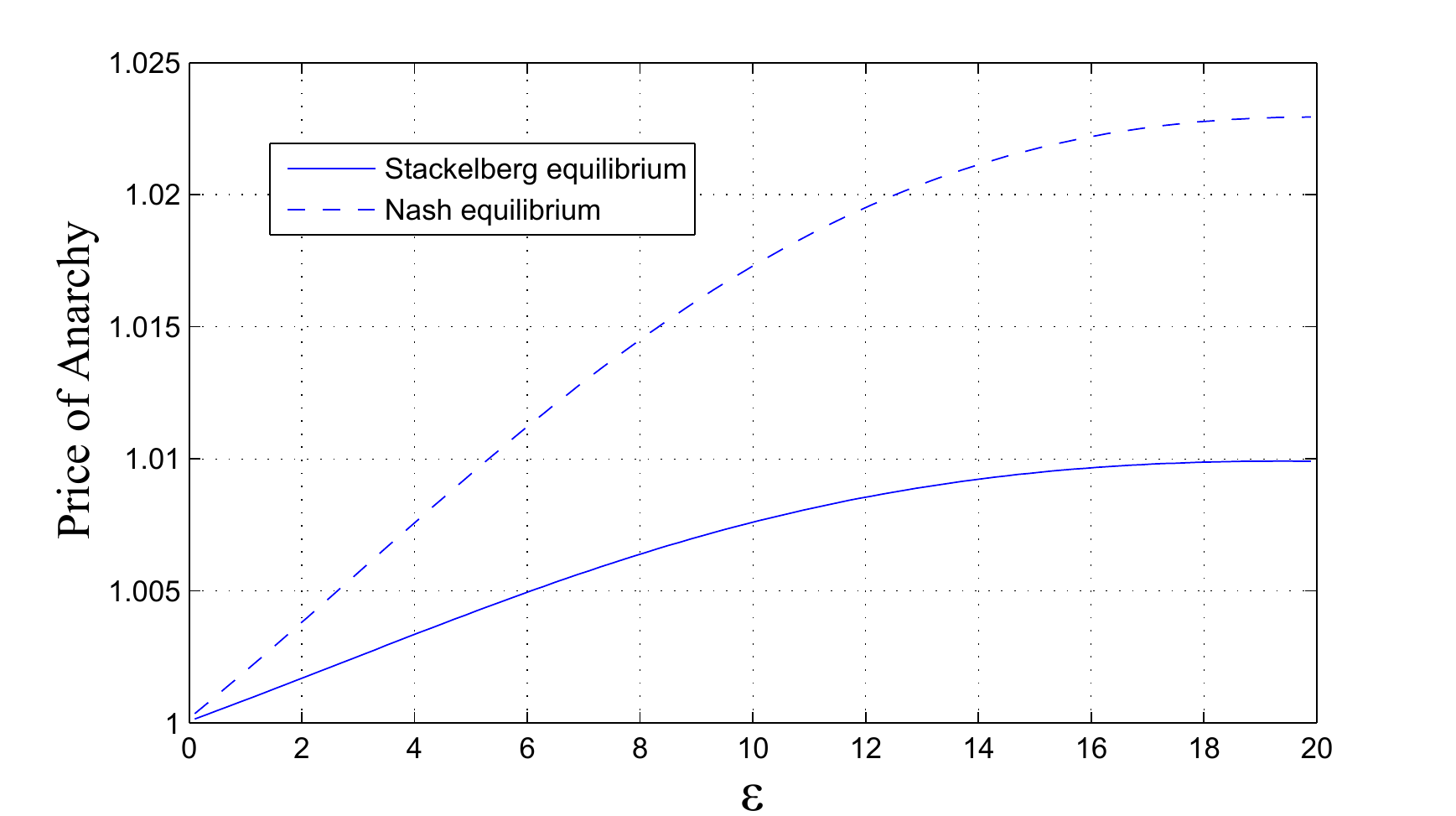}
\caption{Prices of Anarchy regarding to $\varepsilon$. Parameters are $L=100$%
, $\sigma^2 = 10^4$, $\alpha = 2$.}
\label{PoA}
\end{figure}

\section{Convergence to the Nash equilibrium}

\label{sec:dynamic}

Getting to the equilibrium given in~Section~\ref{sec:nash_carac} is not
simple: it would require that every player knows the utilities of other
players. In practice, this is hardly the case. Thus, we present techniques
that enable to reach the Nash equilibrium in a decentralized way:
best-response dynamics and reinforcement learning. However the BSs need to
be movable to perform these two techniques, hence it is more accurate to
talk about Mobile Stations (MSs) in the present section. Note that for these two techniques, the only assumption about MSs location is that they cannot superimpose.

\subsection{Best response dynamics}

\label{sec:bestrep}
The context of this section remains the same as defined in Section~\ref{Ssec:GameDef}.

The principle of best-response dynamics is that given a realization of
actions for its opponents, every player of the game is able to compare its
own possible actions and choose which one is best for itself. Precisely the
best-response algorithm is the following

\begin{enumerate}
\item  At every time step $t$, each MS $k$ chooses its location $x^{\text{br}}_{k}(t)$
according to
\[
x^{\text{br}}_{k}(t)=\arg \max_{x_{k}}\,\hat{U}_{k}(x_{k},\underline{x}_{-k}(t)).
\]

\item  Algorithm stops when $\left|\underline{x}^{\text{br}}(t+1)-\underline{x}^{\text{br}}(t)\right| < \beta$, with $\beta$ fixed.
\end{enumerate}

Regarding to the convergence of this algorithm, we study the system of best-responses (\ref{eq:charac}).
If we neglect $\varepsilon$ regarding to $x_1$ and $L-x_{K-1}$, one gets a linear system of equations
\begin{equation}
\left\{
\begin{array}{l}
x_{1}=\frac{2^{\frac{1+\alpha}{\alpha }}-2^{\frac{2}{\alpha }%
}}{4-2^{\frac{2}{\alpha }}}x_{2}, \\
x_{k}=\frac{x_{k+1}+x_{k-1}}{2},\;\forall k\in \{2,\ldots ,K-1\}, \\
x_{K}=\frac{(4-2^{\frac{1+\alpha}{\alpha}})L}{4-2^{\frac{2}{\alpha }}} +\frac{2^{\frac{1+\alpha}{\alpha }}-2^{\frac{2}{\alpha }%
}}{4-2^{\frac{2}{\alpha }}}x_{K-1} .
\end{array}
\right.
\end{equation}
It is a linear Cournot oligopoly~\cite{Cournot1838}.

One distinguishes two ways this algorithm can run.

\emph{Simultaneous best response dynamics.} At every step, every MS
adapt their locations simultaneously. In this case, the evolution of
the locations with the algorithm can be expressed by
\begin{equation}
\underline{x}^{\text{br}}(t+1) = \underline{a} + \text{A}_K\underline{x}^{\text{br}}(t)
\end{equation}
with $\underline{a} = ~^t\bigl(0,\ldots,0,\frac{(4-2^{\frac{1+\alpha}{\alpha}})L}{4-2^{\frac{2}{\alpha }}}\bigg)$ and
\begin{equation}
\text{A}_K = \begin{pmatrix}
0 & g(\alpha) & 0 & \hdotsfor{3} & 0 \\
\frac12 & 0 & \frac12 & \ddots & \; & \; & \vdots \\
0 & \ddots & \ddots & \ddots & \ddots & \; & \vdots \\
\vdots & \ddots & \ddots & \ddots & \ddots & \ddots & \vdots \\
\vdots & \; & \ddots & \ddots & \ddots & \ddots & 0 \\
\vdots & \; & \; & \ddots & \frac12 & 0 & \frac12 \\
0 & \hdotsfor{3} & 0  & g(\alpha) & 0
\end{pmatrix},
\end{equation}
with $g(\alpha) = \frac{2^{\frac{1+\alpha}{\alpha }}-2^{\frac{2}{\alpha }}}{4-2^{\frac{2}{\alpha }}}$.

As recalled in Lemma 2.8 of~\cite{Varga2000}, for an irreducible positive square matrix A$=(a_{ij})_{1\leq i \leq n}$, then either
\begin{equation}
\sum_{j=1}^{n}{a_{ij}} = \rho(\text{A})\; \forall i\in [1,n],
\end{equation}
or
\begin{equation}
\inf_{1\leq i \leq n}  \sum_{j=1}^{n}{a_{ij}} < \rho(\text{A}) < \sup_{1\leq i \leq K}  \sum_{j=1}^{n}{a_{ij}},
\end{equation}
with $\rho(\text{A})$ being the radius of A.
In our case, $\text{A}_K$ can be verified to be positive and irreducible, then one has
\begin{equation}
g(\alpha) < \rho(\text{A}_K) < 1.
\end{equation}
Hence, the convergence of the algorithm is ensured.

\emph{Sequential best response dynamics.} The other case is the
sequential algorithm: at every step, only one MS adapts its
location. Depending on the MS $k$ adapting its location, the algorithm evolves according to

\begin{equation}
\underline{x}^{\text{br}}(t+1) = \underline{a} + \text{A}^k_K\underline{x}^{\text{br}}(t),
\end{equation}
with $\text{A}^k_K$ corresponding to the adaptation of the $k$-th MS. $\text{A}^k_K$ may have several forms. 
\begin{itemize}
\item If the $k$-th MS, $k \in \{2,\ldots,K-1\}$ adapts its location, then
\[
\text{A}^k_K = \begin{pmatrix}
1 & 0 & \hdotsfor{4} & 0 \\
\ddots & \ddots & \ddots & \; & \; & \; & \vdots \\
\vdots & 0 & 1 & 0 & \; & \; & \vdots \\
\vdots & \; & \frac12 & 0 & \frac12 & \; & \vdots \\
\vdots & \; & \; & 0 & 1 & 0 & \vdots \\
\vdots & \; & \; & \; & \ddots & \ddots & \ddots \\
0 & \hdotsfor{4}  & 0 & 1 
\end{pmatrix}
\]
In this case, $\text{A}^k_K$ is irreducible and positive and
\begin{equation}
\rho(\text{A}^k_K) = 1.
\end{equation}
\item If the first MS adapts its location, the matrix has the form
\[
\text{A}^1_K = \begin{pmatrix}
0 & g(\alpha) & 0 & \hdots \\
0 & 1 & 0 &  \; \\
\vdots & \ddots & \ddots & \ddots \\
0 & \hdots & 0 & 1 
\end{pmatrix}
\]
Again, $\text{A}^1_K$ is irreducible and positive, but this time
\begin{equation}
g(\alpha) < \rho(\text{A}^1_K) < 1.
\end{equation}
Note that if the $K$-th MS adapts its location, the same reasoning can be done.
\begin{equation}
g(\alpha) < \rho(\text{A}^K_K) < 1.
\end{equation}

\end{itemize}
Hence, $\rho(\prod_{k=1}^{K}{\text{A}^k_K}) = \prod_{k=1}^{K}{\rho(\text{A}^k_K)} < 1 $ and the algorithm converges.

In Figure~\ref{best_responses}, we compare MSs $1$ and $2$ sequential best-responses for $\alpha = 2$.

\begin{figure}[tbp]
\hspace{-0.7cm}
\includegraphics[scale=0.45]{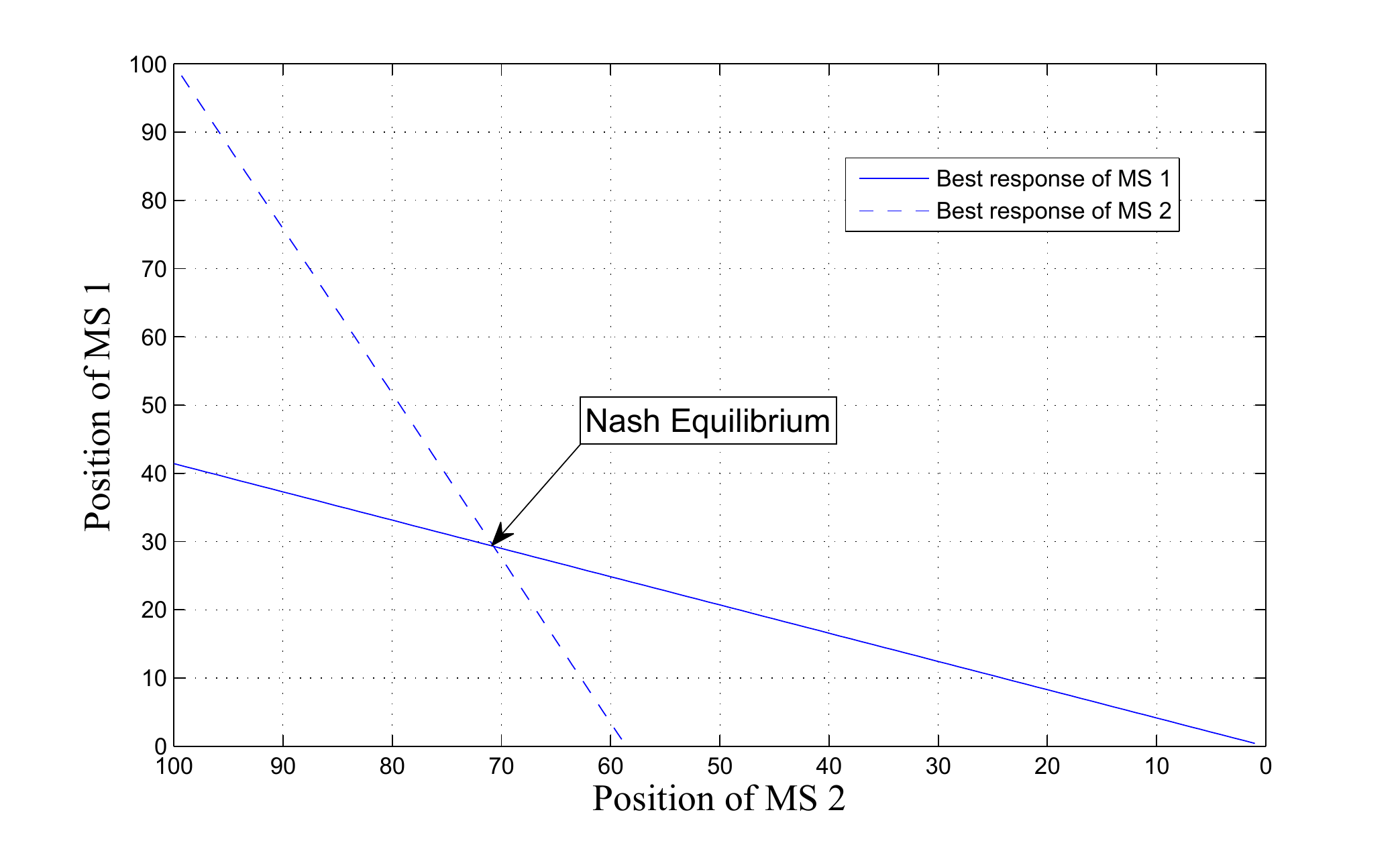}
\caption{Best-responses of MS $1$ and MS $2$ for $\alpha = 2$}
\label{best_responses}
\end{figure}

\subsection{Reinforcement Learning}

\label{sec:learning}
The following notations are used in this section. For a set of MSs $\mathcal{K}=\{1,\ldots ,K\}$, let $%
\mathcal{\tilde{A}}_{k}=\{y_{k1},\ldots ,y_{km_{k}}\}$ be the set of possible
locations for MS $k$ ($m_{k}$ being the cardinality of $\mathcal{\tilde{A}}_{k}$),
which corresponds to a discretization of the set $\left[ 0,L\right] $.

We implement a discrete stochastic
learning algorithm in the sense that the action space is discrete and at
every step, actions are chosen in a stochastic way. The only information
available to a MS is the value of its utility after each iteration (note
that the MS does not necessarily knows its utility expression).

We define $\underline{p}_{k}(t)=(p_{k1}(t),\ldots ,p_{km_{k}}(t))$ the
probability distribution vector of MS $k$ at time $t$.
\begin{equation}
\text{P}\left[ x_{k}(t)=y_{ki}\right] =p_{ki}(t),\;i\in \{1,\ldots ,m_{k}\}.
\end{equation}

The algorithm used by each MS is then the following.

\begin{enumerate}
\item  Initialize the distribution probability vector.

$\forall k\in \mathcal{K},\;\forall i\in \{1,\ldots ,m_{k}\},\;p_{ki}(0)=%
\frac{1}{m_{k}}.$

\item  At every step $t$, each MS $k$ chooses a location $x_{k}(t)$
according to its probability vector $\underline{p}_{k}(t)$.

\item  Each MS gets $U_{k}(t)$. 

\item  Each MS updates its probability distribution vector $\underline{p}%
_{k}(t)$
\begin{equation}
\begin{array}{ll}
p_{ki}(t+1)=p_{ki}(t)-bU_{k}(t)p_{ki}(t)\text{, if }x_{t}(t)\neq y_{ki} &
\\
p_{ki}(t+1)=p_{ki}(t)+bU_{k}(t)\sum_{s\neq i}{p_{k}s}(t)\text{, if }%
x_{k}(t)=y_{ki} &
\end{array}
\end{equation}

\item  Algorithm stops when $\underline{p}_{k}(t+1)=\underline{p}_{k}(t)$,
else go to step $2$.
\end{enumerate}

\begin{figure}[tbp]
\hspace{-0.5cm} \includegraphics[scale=0.55]{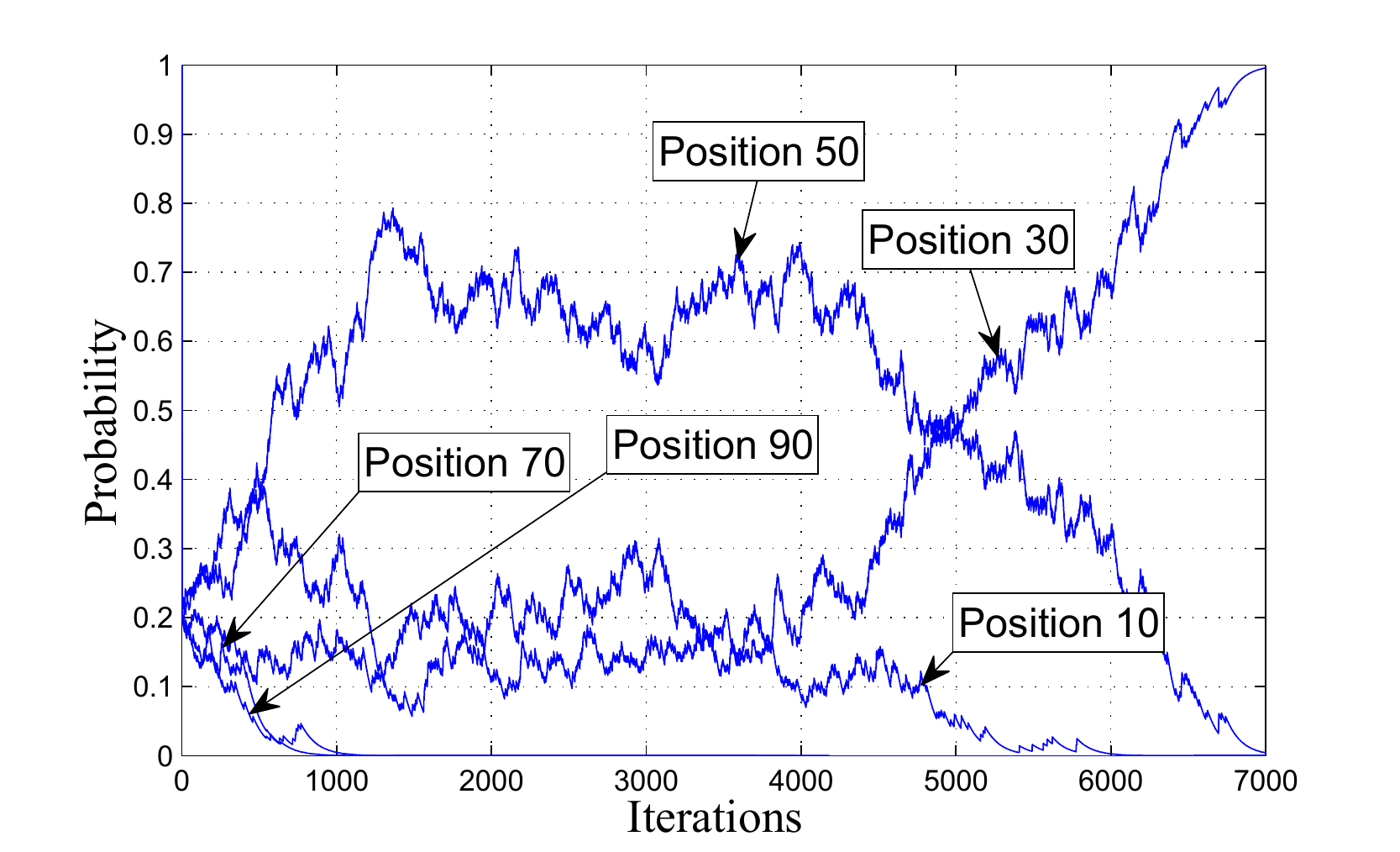}
\caption{Evolution of action probabilities for BS 1 with discrete stochastic
learning.}
\label{fig:learning_joueur1}
\end{figure}

\begin{figure}[tbp]
\hspace{-0.5cm} \includegraphics[scale=0.55]{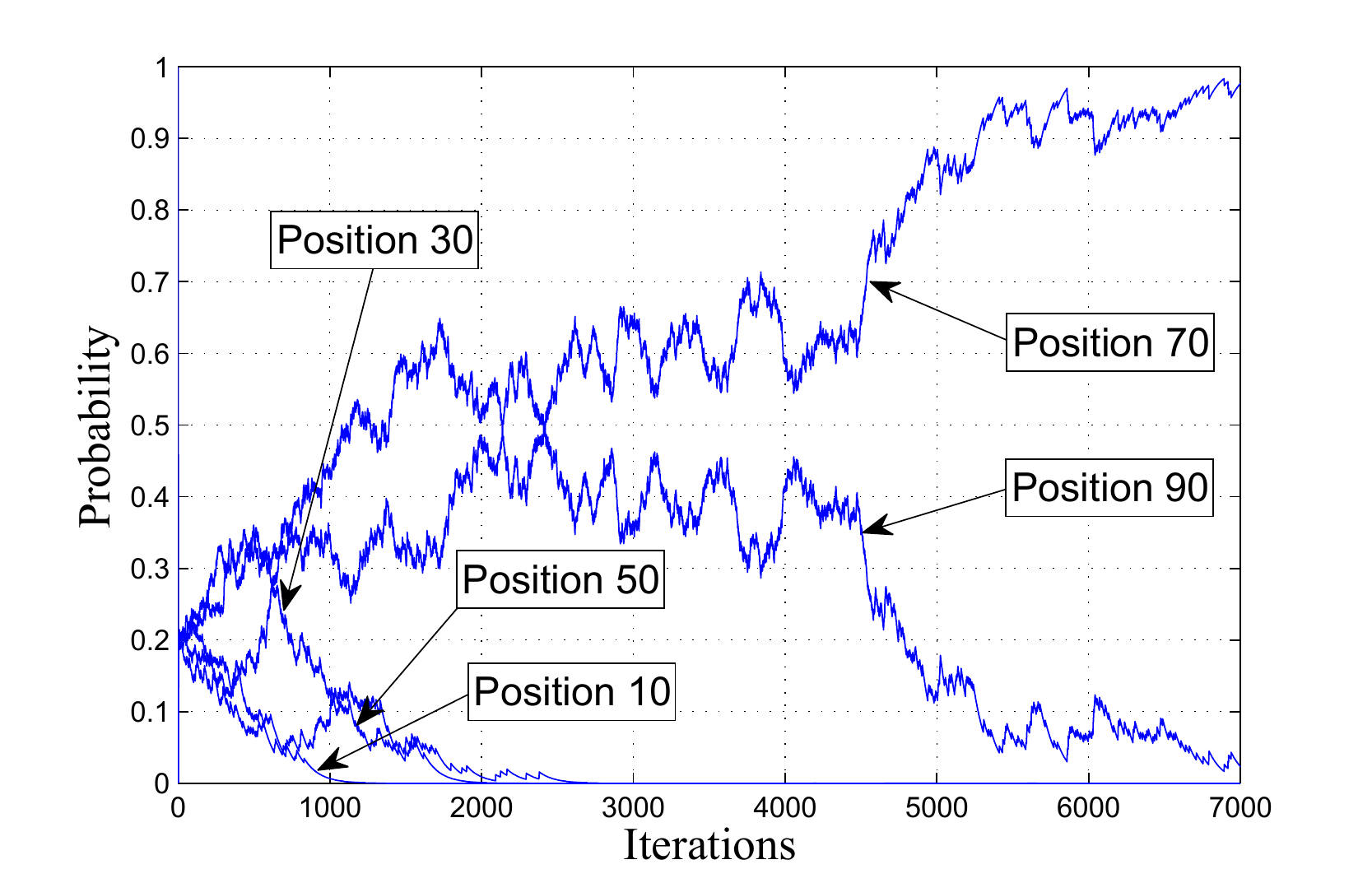}
\caption{Evolution of action probabilities for BS 2 with discrete stochastic
learning.}
\label{fig:learning_joueur2}
\end{figure}

Figures~\ref{fig:learning_joueur1} and~\ref{fig:learning_joueur2} illustrate
the evolution of the probability distribution vectors of two MSs. The
parameters of the simulation are the following: each MS has the same set of
possible positions $\{10,30,50,70,90\}$ and the step of the learning
algorithm is $b=0.01$.

Depending on the choice of $b$, two phenomena occur.

\begin{itemize}
\item  The higher the value of $b$, the lower the convergence time of the
algorithm,

\item  However, if $b$ is chosen too high, the algorithm may converge to
locations that do no correspond to an equilibrium of the game defined in
Section~\ref{Ssec:GameDef}.
\end{itemize}

Figure~\ref{fig:step} illustrates the convergence time as a function of $b$.
\begin{figure}[tbp]
\hspace{-0.5cm} \includegraphics[scale=0.45]{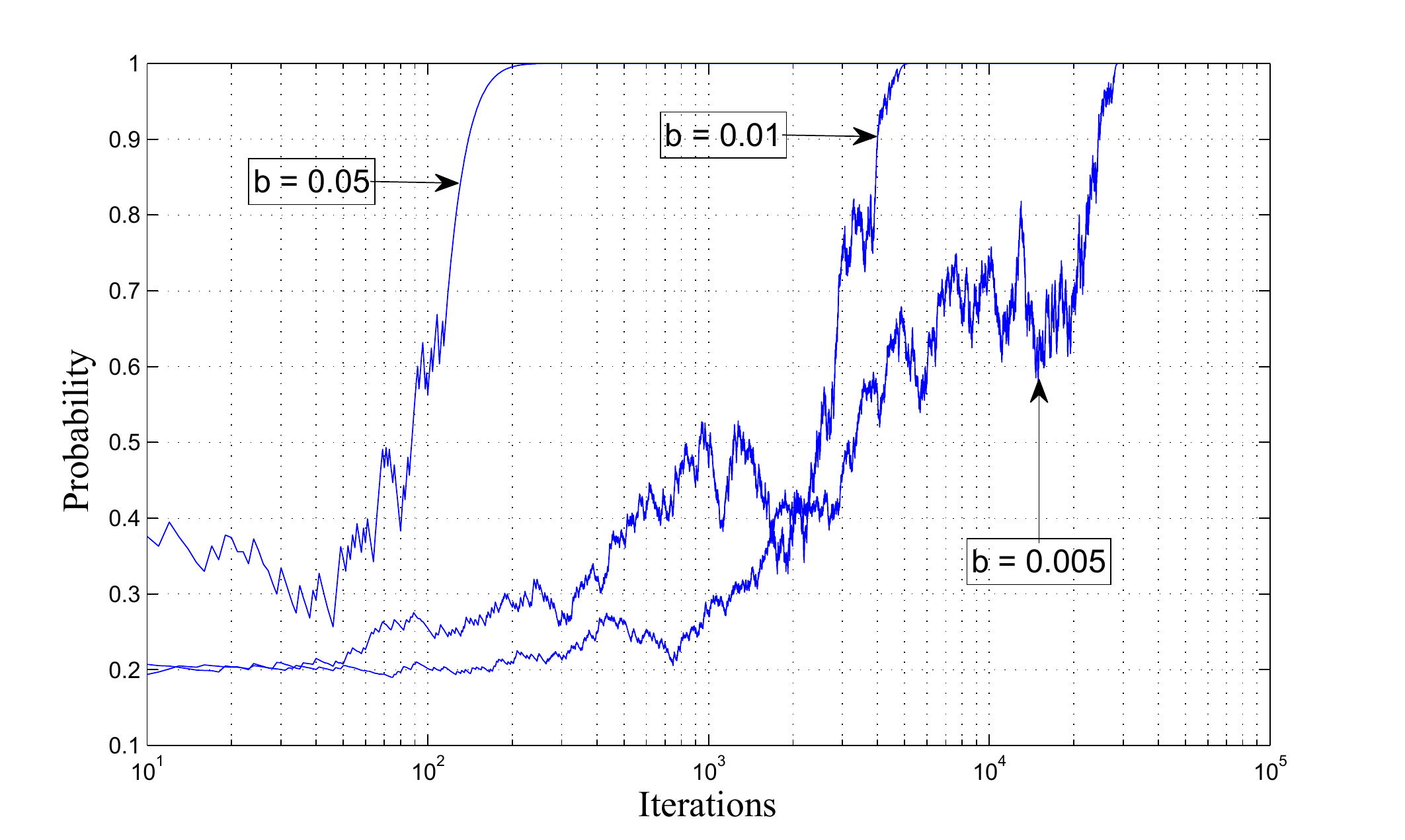} 
\caption{Convergence of probability with respect to step $b$}
\label{fig:step}
\end{figure}

The choice of $b$ is hence a trade-off between convergence time and accuracy
of the convergence.


\section{Conclusion and perspectives}
\label{sec:concl}

Obviously the proposed model is relatively simple and should be
improved to obtain refined results. However, this led us to several
interesting results such as a full characterization of Nash
equilibria and interesting behavior in terms of convergence. It
would be very relevant to extend this work to two-dimensional
scenario, define a suitable order for which uniqueness would be
ensured. More connections with the famous multi-source Weber problem~\cite{SandorContinuousWeber2003}
should be established to better understand the general problem of
\emph{deployment games}. Indeed, if an operator has to locate a set
of base stations, the problem becomes more complicated. The problem
becomes even more interesting if the set of possible constellations
is discrete, which would lead us to make connections with Voronoi
games~\cite{DurrVoronoi2007}. The authors believe there is large avenue for
contributions to the general problem under study, especially in
finding relevant assumptions to simplify it without too loss in
terms of understanding. Random matrix theory and stochastic geometry
might be for great help to achieve this challenging objective.

\bibliographystyle{plain}
\bibliography{md}

\newpage
\appendix
\section{Proofs}
\subsection{Proof of Lemma~\ref{lemma:concavity}}
\label{proof:concavity}

We need to prove that $\widehat{U_k}(\underline{x})$ is concave with respect to $%
x_{k}$ over $\mathcal{A}_k$.

We know that the index of players verify $\left(
\ref{eq:order}\right) $. In this case, the regions associated to the BSs are
\begin{equation}
\left\{
\begin{array}{l}
\text{\lbrack }0,\frac{x_{1}+x_{2}}{2}\text{] for BS $1,$} \\
\text{\lbrack }\frac{x_{k-1}+x_{k}}{2},\frac{x_{k}+x_{k+1}}{2}\text{]}\text{
for BS $k$, }k\in \{2,\ldots ,K-1\} \\
\text{\lbrack }\frac{x_{K-1}+x_{K}}{2},L\text{]}\text{ for BS $K.$}
\end{array}
\right.   \label{eq:regions}
\end{equation}
Then
\begin{equation}
\left\{
\begin{array}{l}
\widehat{U_1}(\underline{x})=\int_{0}^{\frac{x_{1}+x_{2}}{2}}(\varepsilon
^{2}+(x-x_{1})^{2})^{-\frac{\alpha }{2}}{\text{d}x} \\
\widehat{U_k}(\underline{x})=\int_{\frac{x_{k-1}+x_{k}}{2}}^{\frac{x_{k}+x_{k+1}}{2}%
}(\varepsilon ^{2}+(x-x_{k})^{2})^{-\frac{\alpha }{2}}{\text{d}x},\;k\in \{2,\ldots
,K-1\} \\
\widehat{U_K}(\underline{x})=\int_{\frac{x_{K-1}+x_{K}}{2}}^{L}(\varepsilon
^{2}+(x-x_{K})^{2})^{-\frac{\alpha }{2}}{\text{d}x}
\end{array}
\right.
\end{equation}
which may be rewritten as
\begin{equation}
\left\{
\begin{array}{lll}
\widehat{U_1}(\underline{x})=\int_{-x_{1}}^{\frac{x_{2}-x_{1}}{2}}(\varepsilon
^{2}+x^{2})^{-\frac{\alpha }{2}}{\text{d}x} &  &  \\
\widehat{U_k}(\underline{x})=\int_{\frac{x_{k-1}-x_{k}}{2}}^{\frac{x_{k+1}-x_{k}}{2}%
}(\varepsilon ^{2}+x^{2})^{-\frac{\alpha }{2}}{\text{d}x},\;k\in \{2,\ldots ,K-1\}
&  &  \\
\widehat{U_K}(\underline{x})=\int_{\frac{x_{K-1}-x_{K}}{2}}^{L-x_{K}}(\varepsilon
^{2}+x^{2})^{-\frac{\alpha }{2}}{\text{d}x} &  &
\end{array}
\right.
\end{equation}

To prove the existence of a Nash equilibrium, the concavity of $\widehat{U_k}\left(
\underline{x}\right) $ with respect to $x_{k}$ $\forall k\in \mathcal{K}$
has now to be established. One has the first-order partial derivatives
\begin{equation}
\left\{
\begin{array}{l}
\frac{\partial \widehat{U_1}}{\partial x_{1}}(\underline{x})=-\frac{1}{2}%
(\varepsilon ^{2}+(\frac{x_{2}-x_{1}}{2})^{2})^{-\frac{\alpha }{2}%
}+(\varepsilon ^{2}+x_{1}^{2})^{-\frac{\alpha }{2}} \\
\frac{\partial \widehat{U_k}}{\partial x_{k}}(\underline{x})=-\frac{1}{2}%
(\varepsilon ^{2}+(\frac{x_{k+1}-x_{k}}{2})^{2})^{-\frac{\alpha }{2}}+\frac{1%
}{2}(\varepsilon ^{2}+(\frac{x_{k-1}-x_{k}}{2})^{2})^{-\frac{\alpha }{2}} \\
\frac{\partial \widehat{U_K}}{\partial x_{K}}(\underline{x})=-(\varepsilon
^{2}+(L-x_{K})^{2})^{-\frac{\alpha }{2}}+\frac{1}{2}(\varepsilon ^{2}+(\frac{%
x_{K-1}-x_{K}}{2})^{2})^{-\frac{\alpha }{2}}
\end{array}
\right.  \label{eq:deriv}
\end{equation}
and the second-order partial derivatives
\begin{equation}
\left\{
\begin{array}{l}
\frac{\partial ^{2}\widehat{U_1}}{\partial x_{1}^{2}}(\underline{x})=-\frac{\alpha }{%
8}\frac{x_{2}-x_{1}}{(\varepsilon ^{2}+(\frac{x_{2}-x_{1}}{2})^{2})^{\frac{%
\alpha }{2}+1}}-\alpha \frac{x_{1}}{(\varepsilon ^{2}+x_{1}^{2})^{\frac{%
\alpha }{2}+1}} \\
\frac{\partial ^{2}\widehat{U_k}}{\partial x_{k}^{2}}(\underline{x})=\frac{\alpha }{8%
}\bigl(-\frac{x_{k+1}-x_{k}}{(\varepsilon ^{2}+(\frac{x_{k+1}-x_{k}}{2}%
)^{2})^{\frac{\alpha }{2}+1}}+\frac{x_{k-1}-x_{k}}{(\varepsilon ^{2}+(\frac{%
x_{k-1}-x_{k}}{2})^{2})^{\frac{\alpha }{2}+1}}\bigr) \\
\frac{\partial ^{2}\widehat{U_K}}{\partial x_{K}^{2}}(\underline{x})=-\alpha \frac{%
L-x_{K}}{(\varepsilon ^{2}+(L-x_{K})^{2})^{\frac{\alpha }{2}+1}}+\frac{%
\alpha }{8}\frac{x_{K-1}-x_{K}}{(\varepsilon ^{2}+(\frac{x_{K-1}-x_{K}}{2}%
)^{2})^{\frac{\alpha }{2}+1}}
\end{array}
\right.
\end{equation}

Given $\left( \ref{eq:order}\right) $, we have $\frac{\partial ^{2}\widehat{U_k}}{%
\partial x_{k}^{2}}(\underline{x})<0\;\forall k\in \mathcal{K}$. Thus $\widehat{U_k}(%
\underline{x})$ is concave with respect to $x_{k}$ over$\mathcal{A}_k$, $\forall k\in \mathcal{K}
$.


\subsection{Proof of Theorem~\ref{th:unicity}}
\label{proof:unicity}

In the context of $K$-player game, the DSC~\cite{Rosen1965} condition writes $\forall(\underline{a},\underline{a}') \in \mathcal{A}^2$ such that $\underline{a}\neq \underline{a}'$
\begin{equation}
\label{eq:DCS_condition}
\sum_{k=1}^K(a'_k-a_k)\biggl(\frac{\partial \widehat{U_k}}{\partial x_k}(\underline{a})-\frac{\partial \widehat{U_k}}{\partial x_k}(\underline{a}')\biggr)> 0
\end{equation}

For clarity reasons, we denote
\begin{equation}
g(a,b)=\biggl((\epsilon^2+a^2)^{-\frac{\alpha}{2}} - (\epsilon^2+b^2)^{-\frac{\alpha}{2}}\biggr),\;(a,b)\in\mathbb{R}^2.
\end{equation}
By (\ref{eq:deriv}), it turns
\begin{equation}
\left\{
\begin{array}{l}
\scriptstyle{\frac{\partial \widehat{U_1}}{\partial x_1}(\underline{a})-\frac{\partial \widehat{U_1}}{\partial x_1}(\underline{a}')=
g(a_1,a'_1) - \frac{1}{2}\biggl(g(\frac{a_{1}-a_{2}}{2},\frac{a'_{1}-a'_{2}}{2})\biggr)} \\
\scriptstyle{\frac{\partial \widehat{U_k}}{\partial x_k}(\underline{a})-\frac{\partial \widehat{U_k}}{\partial x_k}(\underline{a}')=
\frac{1}{2}\biggl(g(\frac{a_{k-1}-a_k}{2},\frac{a'_{k-1}-a'_k}{2}) - g(\frac{a_{k}-a_{k+1}}{2},\frac{a'_{k}-a'_{k+1}}{2})\biggr)}\\
\scriptstyle{\frac{\partial \widehat{U_K}}{\partial x_K}(\underline{a})-\frac{\partial \widehat{U_K}}{\partial x_K}(\underline{a}')=
\frac{1}{2}\biggl(g(\frac{a_{K-1}-a_K}{2},\frac{a'_{K-1}-a'_k}{2})\biggr) - g(L-a_K,L-a'_K)}
\end{array}
\right.  
\end{equation}

Equation (\ref{eq:DCS_condition}) becomes
\begin{equation}
\begin{aligned}
&\scriptstyle{(a'_1-a_1)\biggl(g(a_1,a'_1) - \frac{1}{2}g(\frac{a_{1}-a_{2}}{2},\frac{a'_{1}-a'_{2}}{2})\biggr)} \\
&\scriptstyle{+\sum_{k=2}^{K-1}\frac{a'_k-a_k}{2}\biggl(g(\frac{a_{k-1}-a_i}{2},\frac{a'_{k-1}-a'_k}{2}) - g(\frac{a_{k}-a_{k+1}}{2},\frac{a'_{k}-a'_{k+1}}{2})\biggr)} \\
&\scriptstyle{+(a'_K-a_K)\biggl(\frac{1}{2}g(\frac{a_{K-1}-a_K}{2},\frac{a'_{K-1}-a'_K}{2}) - g(L-a_K,L-a'_K)\biggr) > 0}
\end{aligned}
\end{equation}
which can also be written
\begin{equation}
\begin{aligned}
&(a'_1-a_1)g(a_1,a'_1) \\
&+\sum_{k=2}^{K}\frac{a'_{k}-a'_{k-1}-(a_{k}-a_{k-1})}{2}g\biggl(\frac{a_{k-1}-a_k}{2},\frac{a'_{k-1}-a'_k}{2}\biggr) \\
&+(L-a'_K-(L-a_K))g(L-a_K,L-a'_K) > 0
\end{aligned}
\end{equation}
However $\forall (a,b) \in \mathbb{R}^{*2}$, 
\begin{equation}
(b-a)g(a,b) > 0
\end{equation}
and by the order condition~(\ref{eq:order})
\begin{equation}
\left\{
\begin{array}{lll}
(a'_1,a_1) \in \mathbb{R}^{*2} \\
(\frac{a'_{k}-a'_{k-1}}{2},\frac{a_{k}-a_{k-1}}{2}) \in \mathbb{R}^{*2}\; \forall k \in \{2,\ldots,K\}\\
(L-a'_K,L-a_K) \in \mathbb{R}^{*2}
\end{array}
\right.
\end{equation}

Then the DSC condition is verified and the equilibrium is unique.

\subsection{Derivation of (\ref{eq:charac})}
\label{proof:charac} To obtain a formal expression of a Nash equilibrium,
the intersection of the best-responses has to be considered
\begin{equation}
\frac{\partial \hat{U}_{k}}{\partial x_{k}}(\underline{x})=0,\;\forall k\in
\mathcal{K},\;\forall \alpha \geqslant 2.  \label{Eq:BestRespInter}
\end{equation}

$\left( \ref{Eq:BestRespInter}\right) $
leads to
\begin{equation}
\left\{
\begin{array}{l}
(2^{\frac{2}{\alpha }}(\varepsilon ^{2}+(\frac{x_{2}-x_{1}}{2})^{2}))^{\frac{%
\alpha }{2}}=(\varepsilon ^{2}+x_{1}^{2})^{\frac{\alpha }{2}} \\
(\varepsilon ^{2}+(\frac{x_{k+1}-x_{k}}{2})^{2})^{\frac{\alpha }{2}%
}=(\varepsilon ^{2}+(\frac{x_{k-1}-x_{k}}{2})^{2})^{\frac{\alpha }{2}%
},\;\forall k\in \{2,\ldots ,K-1\} \\
(\varepsilon ^{2}+(L-x_{K})^{2})^{\frac{\alpha }{2}}=(2^{\frac{2}{\alpha }%
}(\varepsilon ^{2}+(\frac{x_{K-1}-x_{K}}{2})^{2}))^{\frac{\alpha }{2}}.
\end{array}
\right.
\end{equation}
Since all terms elevated at power $\alpha /2$ are positive, one gets
\begin{equation}
\left\{
\begin{array}{l}
2^{\frac{2}{\alpha }}(\varepsilon ^{2}+(\frac{x_{2}-x_{1}}{2}%
)^{2})=\varepsilon ^{2}+x_{1}^{2} \\
(\frac{x_{i+k}-x_{k}}{2})^{2}=(\frac{x_{k-1}-x_{k}}{2})^{2},\;\forall k\in
\{2,\ldots ,K-1\} \\
2^{\frac{2}{\alpha }}(\varepsilon ^{2}+(\frac{x_{K-1}-x_{K}}{2}%
)^{2})=\varepsilon ^{2}+(L-x_{K})^{2}
\end{array}
\right.
\end{equation}

One real solution $\underline{x}$ verifying $\left( \ref{eq:order}\right) $
has thus to satisfy $\left( \ref{eq:charac}\right) $ when$\,\alpha >2$ and $%
\left( \ref{eq:charac2}\right) $ when $\alpha =2$.

\end{document}